\documentclass[12pt,reqno]{amsart}
\usepackage{fullpage,amsfonts, amssymb, amsmath, enumerate, mathtools,amsthm,amstext,graphicx}

\usepackage{url}

\newcommand{\p}[1]{\mathbb{P}\left[{#1}\right]}
\newcommand{\e}[1]{\mathbb{E}\left[{#1}\right]}
\newcommand{\indeg}{\operatorname{deg}^{-}}

\newcommand{\eq}[1]{(\ref{#1})}
\newcommand{\aas}{{\sl a.a.s.}\ }
\newcommand{\eps}{\varepsilon}
\newtheorem{theorem}{Theorem}[section]

\newtheorem{corollary}[theorem]{Corollary}
\newtheorem{lemma}[theorem]{Lemma}
\newtheorem{conjecture}[theorem]{Conjecture}

\newtheorem{definition}[theorem]{Definition}
\newtheorem{question}[theorem]{Question}

\theoremstyle{remark}
\newtheorem*{remark}{Remark}

\begin{document}
\author{Jeannette Janssen}
\address{JJ:
Department of Mathematics and Statistics\\
Dalhousie University\\
6316 Coburg Road\\
PO BOX 15000\\
Halifax, Nova Scotia\\
Canada B3H 4R2}
%Department of Mathematics and Statistics\\
%Dalhousie University\\
%Halifax, NS\\
%B3H 3J5 Canada}
\email{jeannette.janssen@dal.ca}
\thanks{The collaboration between the authors was a result of the visit of the first author to the Institute of Mathematics and Applications (IMA) in Minnesota. She wishes to thank IMA for providing this opportunity. 
She also acknowledges NSERC for their support of this research.} 
\author{Abbas Mehrabian}
\address{AM: University of Waterloo, Waterloo, ON, Canada}
\curraddr{Department of Computer Science\\
University of British Columbia\\
2366 Main Mall\\
Vancouver, B.C.\\ 
V6T 1Z4 Canada}
\email{AbbasMehrabian@gmail.com}
\urladdr{http://www.cs.ubc.ca/~amehrabi/}
\thanks{The second author was supported by the Vanier Canada Graduate Scholarships program,
a PIMS Postdoctoral Fellowship, and an
NSERC Postdoctoral Fellowship}

\date{\today}
\title{Rumours spread slowly in a small world spatial network}
\thanks{A preliminary version of this paper appeared in proceedings of the 12th workshop on algorithms and models for the web graph, WAW 2015, published by Springer in the Lecture Notes for Computer Science 9479, pp.\ 107--118.
In this full version, all the proofs are included, and the upper bound for the effective diameter is improved to $O(\log^2 n)$.}


\begin{abstract}
Rumour spreading is a protocol for modelling the spread of information through a network via user-to-user interaction. 
%The spread time of a graph is the number of rounds needed to spread the rumour to the entire graph. 
The Spatial Preferred Attachment (SPA) model is a random graph model for complex networks: vertices are placed in a metric space, and the link probability depends on the metric distance between vertices, and on their degree. 
We show that the SPA model typically produces graphs that have small effective diameter, i.e.\ $O(\log^2 n)$, while rumour spreading is relatively slow, namely polynomial in $n$.
\end{abstract}

\maketitle

\section{Introduction}

There is increasing consensus in the scientific community that complex networks (e.g.~on-line social networks or citation graphs) can be accurately modelled by spatial random graph models. Spatial random graph models are models where the vertices are located in a metric space, and links are more likely to occur between vertices that are close together in this space. The space can be interpreted as a {\sl feature space}, which models the underlying characteristics of the entities represented by the vertices. Specifically, entities with similar characteristics (for example, users in a social network that share similar interests) will be placed close together in the feature space. Thus the distance between vertices is a measure of affinity, and thus affects the likelihood of the occurrence of a link between these vertices.

An important reason to model real-life networks is to be able, through simulation or theoretical analysis, to study the dynamics of information flow through the network. Several ways to model flow of information through a network have been proposed recently, based on metaphors such as the spread of infection or of fire, or the range of a random walk through the graph \cite{burning,coverPA,KempeKleinTardos,EpidemicNetworks}. Here we focus on a protocol called {\sl rumour spreading}. It differs from the models based on fire or infection in that in each round, the rumour spreads to only one neighbour of each informed vertex. On the other hand, the difference with a random walk approach is that each informed vertex spreads the rumour, and thus we have more of a growing tree of random walks. 

In this paper we study the behaviour of the rumour spreading protocol on graphs produced by the Spatial Preferential Attachment (SPA) model, which is a spatial model that produces sparse power law graphs. We show that, on the one hand, the graph distance between vertices in such a graph tends to be small (polylogarithmic in $n$, the number of vertices), while on the other hand, it takes a long time (polynomial in $n$) to spread the rumour to most of the vertices.

\subsection{The SPA model}

The SPA model is a growing graph model, where one new vertex is added to the graph in each time step. The vertices are chosen from a metric space. Each vertex has a sphere of influence, whose size grows with the degree of the vertex.  A new vertex can only link to an existing vertex if it falls inside its sphere of influence. Therefore, links between vertices depend on their (spatial) distance, and on the in-degree of the older vertex. 

The SPA model was introduced in \cite{spa_def}, where it was shown that asymptotically, graphs produced by the SPA model have a power law degree distribution with exponent in $[2,\infty)$ depending on the parameters. The model was further studied in \cite{spa_typical,geoSPA,lumpySPA}. The model can be seen as a special case of the spatial model introduced by Jacob and M\"{o}rters in \cite{jacob1} and further studied in \cite{jacob2}. 
The SPA model has similarities with the spatial models introduced in \cite{geoprotean,threshbrad,geopref2,geoprefZuev}. 

Let $S$ be the unit hypercube in $\mathbb{R}^m$, equipped with the torus metric derived from the Euclidean norm. 
%In this paper, we will mostly consider the case where $m=2$.
%In rumour spreading we consider general m
The SPA model stochastically generates a graph sequence $\{G_t\}_{t\geq 0}$; for each $t \ge 0$, $G_t=(V_t, E_t)$, where $E_t$ is an edge set, and $V_t \subseteq S$ is a vertex set.  The index $t$ is an indication of time.   The in-degree, out-degree and total degree of a vertex $v$ at time $t$ is denoted by $\deg^-(v,t)$, $\deg^+(v,t)$ and
$\deg(v,t)$, respectively.

We now define the {\sl sphere of influence} $S(v,t)$ of a vertex $v$ at time $t$. Let 
$$
A(v,t) \coloneqq \frac{A_1 \deg^-(v,t) + A_2}{t},
$$
where $A_1, A_2>0$ are given parameters.
If $A(v,t)\leq 1$, then $S(v,t)$ is defined as the ball, centred at $v$, with total volume $A(v,t)$. If $A(v,t)>1$ then  $S(v,t) = S$, and so $|S(v,t)| = 1$. To keep the second option from happening often, we impose the additional restriction that $A_1 < 1$; this ensures that in the long run, $S(v,t)\ll 1$ for all $v$. 

The generation of a SPA model graph begins at time $t = 0$ with $G_0$ being the null graph. At each time step $t \geq 1$, a node $v_t$ is chosen from $S$ according to the uniform distribution, and added to $V_{t-1}$ to form $V_t$.  Next, independently for each vertex $u\in V_{t-1}$ such that $v_t \in S(u,t)$, a directed link $(v_{t},u)$ is created with probability $p$. 
%(In the original SPA model, such a link was added with probability $p$; here we consider the special case where $p=1$. We suspect that a different value of $p$ would give similar results.)
%\abas{for rumour spreading we consider $p<1$ as well!}

Because the volume of the sphere of influence of a vertex is proportional to its in-degree, so is the probability of the vertex receiving a new link at a given time. Thus link formation is governed by a preferential attachment, or ``rich get richer'', principle, which leads to a power law degree distribution of the in-degrees, and thus also of sizes of the spheres of influence.  

Another important feature of the model is that all spheres of influence tend to shrink over time. This means that the length of an edge (the distance between its endpoints) depends on the time when it was formed: edges formed in the beginning of the process tend to be much longer than those formed later (see \cite{geoSPA} for more on the distribution of edge lengths). As we will see, the old, long links significantly decrease the graph distance between vertices. This is a feature unique to the SPA model; ``static'' variations of the SPA model such as that presented in \cite{BonGleich}, tend to limit the maximum length of an edge, which leads to a larger diameter.

Note that the SPA model generates directed graphs. However, the rumour spreading protocols we study here completely ignore the edge orientations; we imagine that they work on the corresponding undirected underlying graph. 
Similarly, in estimating the  graph distances, we ignore the edge orientations.

\subsection{Rumour spreading}

Rumour spreading is a model for the spread of one piece of information, the {\sl rumour}, which starts at one vertex, and in each time step, spreads along the edges of the graph according to one of the following protocols. 

The \emph{push protocol} is a round-robin rumour spreading protocol defined as follows: initially one vertex of a simple undirected graph knows a rumour and wants to spread it to all other vertices.
In each round, every informed vertex sends the rumour to a random neighbour.

The \emph{push\&pull protocol} is another round-robin rumour spreading protocol defined as follows: initially one vertex of a simple undirected graph knows a rumour and wants to spread it to all other vertices.
In each round, every informed vertex sends the rumour to a random neighbour, while every uninformed vertex contacts a random neighbour and gets the rumour from her if she knows it.

In both protocols defined above, all vertices work in parallel.
These are synchronized protocols, so if a vertex receives the rumour at round $t$,
it starts passing it on from round $t+1$.
Also, vertices do not have memory,
so a vertex might contact the same neighbour in consecutive rounds.

We are interested in the {\sl spread time}, the number of rounds needed for all vertices to get informed. Since the SPA model does not generally produce connected graphs, we here limit this requirement to vertices in the same component as the starting vertex.
It is clear that the push\&pull protocol is generally quicker (this can be made precise via a coupling argument).

The push protocol was defined in~\cite{push_first} for the complete graph,
and was studied in~\cite{FPRU90} for general graphs.
The push\&pull protocol was defined in~\cite{DGH+87}, where experimental results were presented, and the first analytical results appeared in~\cite{pushpull_rigorous}.

\subsection{Main results}

Clearly, the diameter of a graph is a lower bound on the spread time, at least for appropriate choices of starting vertex. An easy well known upper bound for spread time is $O(\Delta (\mathrm{diameter} + \log n))$  \cite[Theorem~2.2]{FPRU90}, 
where $\Delta$ denotes the maximum degree.
So in graphs of bounded degree, spread time is largely determined by the diameter.  
Another important factor in rumour spreading is the degree distribution of the graph. 
Vertices of high degree tend to slow down the spread, since only one neighbour of a vertex is contacted in each round. 
SPA model graphs have a power law degree distribution, and the maximum degree is typically $\Omega (n^{A_1})$ (see \cite{spa_def}). %\abas{Theta of?}.

In this paper we prove two main results. First, we show that for most pairs of vertices, the graph distance is polylogarithmic in the number of vertices. 
Thus, SPA model graphs are so-called small worlds. SPA model graphs are generally not connected, and the size and threshold of the giant component are not exactly known. Therefore we state our result in terms of the {\sl effective diameter}, introduced in \cite{forest_fire_journal}. A graph $G$ has effective diameter at most $d$ if, for at least 90\% of all pairs of vertices of $G$ that are connected, their graph distance is at most $d$.
We say an event happens \emph{asymptotically almost surely (a.a.s.)} if its probability approaches 1 as $n$ goes to infinity.
All logarithms are in the natural base in this paper.

Recall that the SPA model has four parameters: 
$m\in \mathbb{Z}_+$ is the dimension,
$A_1,A_2>0$ control the volumes of vertices' spheres of influence, 
and $p\in(0,1]$ is the probability of link formation.

\begin{theorem}
\label{thm:main_diam}
For each choice of $A_1\in [0,1)$, and for large enough choice of $A_2$, \aas a graph produced by the SPA model with parameters $A_1,A_2,$ $p=1$ and $m=2$ has effective diameter $O(\log^2 n)$.
\end{theorem}

\begin{remark}
The constant 90\%
in the definition of effective diameter is somewhat arbitrary. Our arguments yield similar bounds if this is changed to any other constant $1-\varepsilon$ strictly smaller than 100\%.
\end{remark}

As noted before, this result refers to the {\em undirected} diameter. In \cite{spa_typical}, it was shown that \aas any shortest directed path has length $O(\log n)$.  This result does not apply to our situation, since pairs connected by a directed path are a small minority.

Bringmann, Keusch, and Lengler~\cite{Bringmann} proved a polylogarithmic upper bound for the diameter of a spatial random graph model with given expected degrees.
In their model the edges appear independently, and so their result does not apply to our model.

We believe the conclusion of Theorem~\ref{thm:main_diam} is not tight.
We make the following conjecture.

\begin{conjecture}
For each choice of $p,A_1\in (0,1)$, and for large enough choice of $A_2$, \aas the giant component of a graph produced by the SPA model with parameters $A_1,A_2,$ $p$ and $m=2$ has diameter $O(\log n)$.
\end{conjecture}

Our proof for Theorem~\ref{thm:main_diam} is based on two-dimensional objects called crossings, and for extending it to higher dimensions new techniques are required.
We leave this as an open problem.

\begin{question}
Extend Theorem~\ref{thm:main_diam}
to higher dimensions $m>2$.
\end{question}

Our second result illustrates that, in spite of the small world property, \aas rumour spreading with the push\&pull protocol is slow, that is, takes polynomial time in $n$.
%(a  polynomial lower bound for the push\&pull protocol is also given in Corollary~\ref{cor2}).
\begin{theorem}
\label{thm:main_push}
Let $G$ be a graph produced by the SPA model with parameters $A_1,A_2>0,m \in \mathbb{Z}_+$, and assume that $a \coloneqq pA_1<1$.
Define $K \coloneqq (3+a)m+1-a$ and 
let $\alpha < a(1-a)/K$ be a constant.
%$$\alpha < \min \{ pA_1/24, (1-pA_1)/12\}  \in (0,1) \:.$$
If a rumour starts in $G$ from a uniformly random vertex, then a.a.s.\ after $n^{\alpha}$ rounds of the push\&pull protocol, the number of informed vertices is $o(n)$. 
\end{theorem}

\begin{remark}
Note that, since the push protocol is not quicker than the push\&pull protocol, the same lower bound holds for  the push protocol.
\end{remark}

Let us also remark that our diameter result, Theorem~\ref{thm:main_diam}, holds only for dimension 2, and this is because our argument is based on building crossings, whereas our rumour spreading bound holds for all dimensions.

%\begin{remark}
%The main objective of this theorem (and Corollary~\ref{cor2}) is to give {\sl{some}} polynomial lower bound for the number of rounds needed to inform $\Omega(n)$ vertices. 
%In particular, we have not tried to optimize the exponent $\alpha$ here.
%\end{remark}

We can understand Theorem~\ref{thm:main_push} as follows. While SPA model graphs have a backbone of long edges that decrease graph distances between vertices, only old edges are long. Old edges have old endpoints, so the vertices on this backbone are old. Old vertices have high degree, and 
vertices of high degree are slower in spreading the rumour.
So if the rumour travels along long edges, then it will become delayed due to high vertex degree, and if it travels along short edges, it takes many steps to cover the entire space. 

In \cite{jacob1} it was shown that, for certain choices of the parameters, the generalized spatial model by Jacob and M\"{o}rters exhibits a similar mixture of long and short edges. This suggests that our results may be extended to this model; this would be an interesting question to pursue. 

\begin{question}
Can our results be extended to the 
Jacob-M\"{o}rter model~\cite{jacob1}?
\end{question}

The push\&pull protocol has been studied on two small-world (non-spatial) models and it turned out that it spreads the rumour in logarithmic time:
it was shown in~\cite{sublogarithmic} that
on a random graph model based on preferential attachment, push\&pull spreads the rumour within $O(\log n)$ rounds.
A similar bound was proved for the performance of this protocol on random graphs with given expected degrees when the average degree distribution is power law~\cite{ultrafast}.
Thus, the SPA model is a unique example of a natural model that exhibits both the small world property and slow rumour spreading.

\section{The effective diameter of SPA model graphs}\label{sec:diameter}

In this section we 
prove Theorem~\ref{thm:main_diam}, which states that a two-dimensional SPA model graph with $p=1$ typically has a small effective diameter. 
Assume that $m=2$ and $p=1$. 
We will derive our bound using properties of the random geometric graph model, especially those studied in \cite{rgg,ganes13,penrose03}.

A two-dimensional random geometric graph on $N$ vertices with radius $r=r(N)$, denoted by $RGG(N,r)$, is generated as follows:
$N$ vertices are chosen independently and uniformly at random from the unit square $S$,
and an edge is added between two vertices if and only if their Euclidean distance is at most $r$.
To see how the geometric random graph model relates to the SPA model, let $\{G_t\}_{t=0}^n$ be a sequence of graphs produced by the SPA model. {Our analysis is based on a sequence of subgraphs of $G_t$, which mimick the behaviour of the model when the in-degree does not influence the size of the sphere of influence.}

For each $t$, define the graph $R_t$ as a
graph with vertex set  $V(G_t)$ in which
two vertices are adjacent if and only if their distance is at most $\sqrt{A_2/t\pi}$.
Observe that $R_t$ conforms to the random geometric graph model on $t$ vertices with radius 
\begin{equation}
r_t := \sqrt{\frac{A_2}{t\pi}} \:. \label{def_rt}
\end{equation}

For all $t$, $R_t$ is a subgraph of (the undirected underlying graph of) $G_t$.
Namely, at all times from 1 to $t$, each sphere of influence has volume at least ${A_2}/{t}$, i.e.\ radius at least $r_t$. Therefore, if two vertices $v_i$ and $v_j$, $1\leq i<j\leq t$, have distance at most $r_t$, then at time $j$, when $v_j$ is born, $v_j$ will fall inside the sphere of influence of $v_i$, and a link $v_jv_i$ will be created. 
We will use the graphs $R_t$ to bound the diameter of $G_n$.

As mentioned earlier, graphs produced by the SPA model are generally not connected. However, we can choose the parameters so that there exists a giant component, i.e.\ a component that contains an $\Omega(1)$ fraction of all vertices. Note that if $R_n$ has a giant component, then so has $G_n$. Moreover, it is known (see \cite{penrose03}) that there exists a constant $a_c$ so that, if $r=\sqrt{\frac{a}{\pi N}}$ with $a >a_{c}$, then \aas  $RGG(N,r)$ has a giant component, while if $a < a_c$ then \aas it does not have one
(note that $a$ is simply the average degree). Experiments give a value of
$a_c \approx 4.51$. %\jj{reference for this? Penrose, $\lambda \pi =a$}. 
Therefore, $G_n$ has a giant component \aas if $A_2>a_c$. It would be interesting to determine whether this value of $A_2$ is indeed the threshold for the emergence of the giant component in $G_n$. Determination of this threshold was left as an open problem in~\cite[Section~5]{spa_typical}.

To show that the  effective diameter of the giant component of $G_n$ is $O(\log^2 n)$, we will proceed as follows. 
Given an arbitrary vertex $v$ in the giant component of $R_n$, and thus of $G_n$, the idea is to find a path of length $O(\log n)$ connecting $v$ to some vertex $y_1$ in the giant component of $R_{n/2}$, then connect $y_1$ to a vertex $y_2$ in the giant  of $R_{n/4}$, and so on. 
It will be more convenient to work not with the giant components but with the so-called \emph{spanning components},
which are defined next.

\begin{definition}[Spanning component]
\label{def:spanning}
Let $R$ be a random geometric graph with parameters $N$ and $r$, and let $M\coloneqq \pi N r^2$.
{Let $W= W(N) \coloneqq Mr \log N$, and assume that $W^{-1}$ is an integer. Partition the unit square $S$ into $W^{-1}$ horizontal 
rectangles of size
$W \times 1$
(the \emph{horizontal slabs}),
and also into
$W^{-1}$ vertical 
rectangles of size
$1\times W$
(the \emph{vertical slabs}).}
For a horizontal slab $L$, a \emph{left-to-right crossing} is a path 
$v_0v_1 \dots v_{k}$
contained in $L$
such that 
$v_0$ has distance $\leq r/5$
to the left side of $L$,
$v_{k}$ has distance $\leq r/5$
to the right side of $L$,
and the distance between $v_i$ and $v_{i+1}$
is $\leq r/2$ for each $0\leq i\leq k-1$.
A \emph{top-to-bottom crossing}
is defined similarly for vertical slabs.
It is easy to see that if $R$ has a left-to-right crossing
for each of the 
$W^{-1}$ horizontal slabs,
and has a top-to-bottom crossing for each of the
$W^{-1}$ vertical slabs,
then the vertices of these crossings
are contained in the same connected component,
which is called a \emph{spanning component}.
\end{definition}
See Figure~\ref{spanning component} for an illustration.
The following lemma 
guarantees the existence of spanning components.
%is a corollary of~\cite[Lemma~2]{ganes13}.
\begin{lemma}
\label{lem:spanning}
There exists an absolute constant $a_{\textnormal{large}}$ such that a random geometric graph $R=RGG(N,r)$ with $\pi N r^2 \geq a_{\textnormal{large}}$
has a spanning component with probability
$\geq 1-O(N^{-1})$.
\end{lemma}
\begin{figure}
\begin{center}
\includegraphics[width=0.5\textwidth]{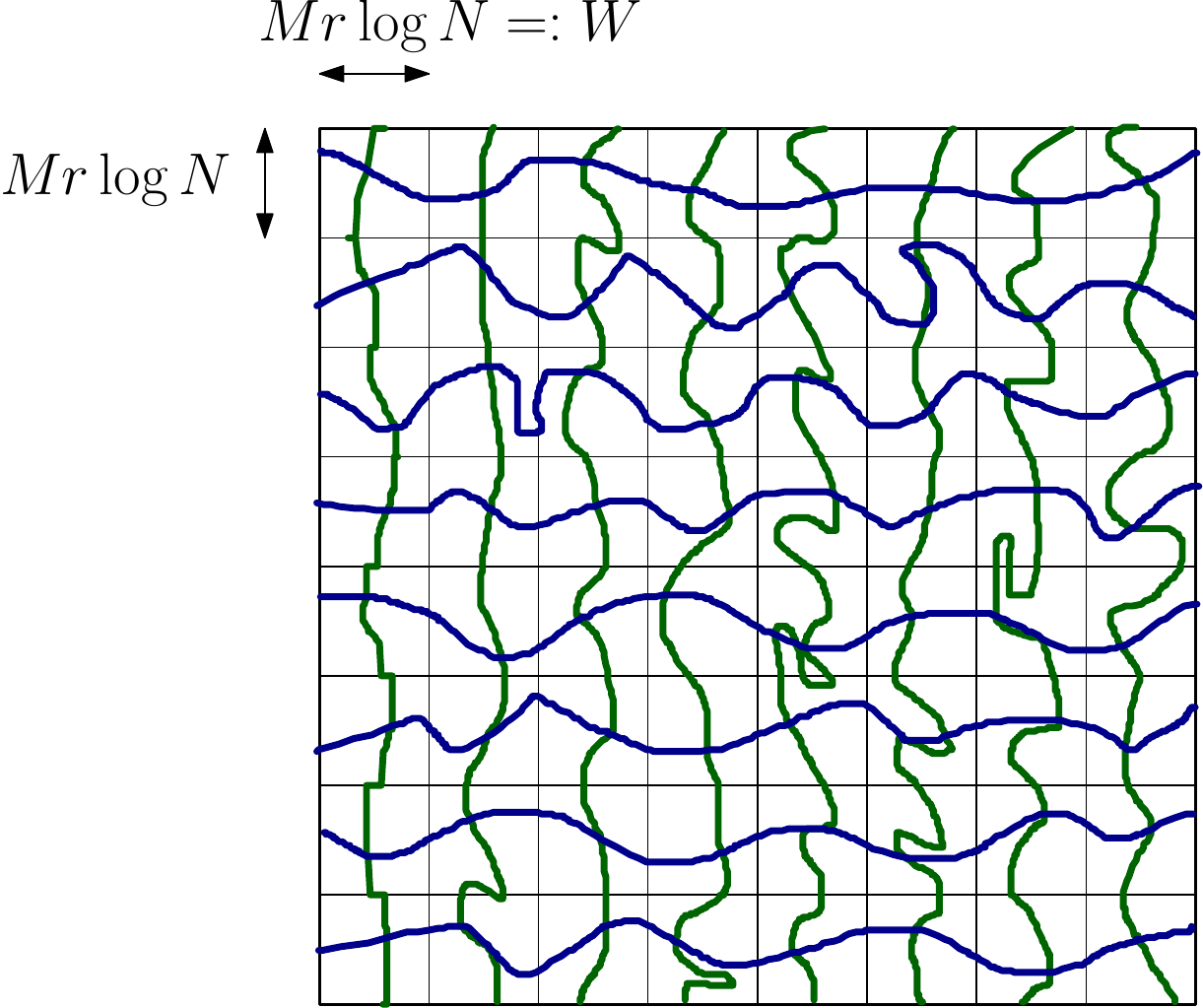}
\caption{A spanning component: left-to-right crossings are blue, top-to-bottom crossings are green.}
\label{spanning component}
\end{center}
\end{figure}
\begin{proof}
Assume that the unit square $S$ is subdivided into $r/5 \times r/5$ subsquares.
Two subsquares are called \emph{adjacent} if they share a side. 
A \emph{path of subsquares} is a sequence of distinct subsquares so that each consecutive pair is adjacent.
A subsquare is called {\em{occupied}} if it contains a vertex of $R$, and {\em{empty}} otherwise. 
For a given rectangle $L$
that is tiled perfectly by subsquares, a \emph{left-to-right subsquare crossing} is a path of occupied subsquares inside $L$ with one endpoint touching the left side of $L$ and the other endpoint touching its right side.
A \emph{top-to-bottom subsquare crossing} is defined similarly.
Note that a left-to-right subsquare crossing 
indeed gives a left-to-right crossing as defined in
Definition~\ref{def:spanning},
by considering the vertices inside the corresponding subsquares
(here we used the fact that, if we have two points lying in side by side squares of side length $r/5$, the distance between the points is less than $r/2$).

Let $M \coloneqq \pi N r^2$.
Ganesan~\cite[Lemma~2]{ganes13}
proved that there exists an absolute constant $c>0$ such that each of the horizontal (vertical) slabs (defined in Definition~\ref{def:spanning}) has a left-to-right (top-to-bottom) subsquare crossing with probability at least 
$1 - O\left(N^{1/2 - c M^2}\right)$.
The total number of slabs is 
$$
2/(Mr \log N)=
O \left( \sqrt N / \log N\right),
$$
so if $M \geq a_{\textnormal{large}}
\coloneqq 2/\sqrt c$,
then by the union bound, with probability at least $1-O\left( N^{-1} \right)$ all horizontal (vertical) slabs have a left-to-right (top-to-bottom) crossing. 
\end{proof}

For two vertices $u$ and $v$ of a random geometric graph $R$,
we denote their Euclidean distance 
and graph distance by 
$d_E(u,v)$ and $d_R(u,v)$, respectively.
The following result was proved by
Friedrich, Sauerwald, and Stauffer
(see Theorem~3 and Remark~5 in~\cite{rgg}).

\begin{theorem}[\cite{rgg}]
\label{lem:rgg}
There exist absolute constants $\Gamma$ and $\eta$ such that a random geometric graph $R=RGG(N,r)$ 
with $\pi N r^2 > 5$
satisfies the following property with probability at least $ 1 - O(N^{-1})$:
for any two vertices $u,v$ that are in the same connected component of $R$
and have
$d_E(u,v) \geq \Gamma (\log N) / (rN)$,
we have
$$
d_R(u,v) \leq \eta d_E(u,v) /  r.
$$
\end{theorem}

\begin{definition}\label{def:nice}
Let $\Gamma$ and $\eta$ be the constants in Theorem~\ref{lem:rgg}.
We say that a random geometric graph $R=RGG(N,r)$ is \emph{nice} if
\begin{itemize}
\item[(i)]
$R$ has a spanning component, and
\item[(ii)]
for any two vertices $u,v$ of $R$
that are connected by a path
and have
$d_E(u,v) \geq \Gamma \log N / (rN)$,
we have
$
d_R(u,v) \leq \eta d_E(u,v) /  r
$.
\end{itemize}
\end{definition}

\begin{lemma}
\label{lem:union}
Suppose that $A_2 \geq \max\{5,a_{\textnormal{large}}\}$.
Let $k$ be the smallest integer such that $n2^{-k}\leq \log n$.
A.a.s.\ we have that 
all random geometric graphs 
$R_n, R_{n/2},\dots, R_{n2^{-k}}$
are nice.
\end{lemma}
\begin{proof}
By the definition of $k$ we have
$n2^{1-k} > \log n$
and so
$2^{k}< 2n/\log n$.
For each $j \in \{n, n/2, \dots, n2^{-k}\}$,
by Lemma~\ref{lem:spanning} and Theorem~\ref{lem:rgg}
and the union bound,
$R_j$ is nice with probability at least
$1 - O(j^{-1})$.
By the union bound, the probability that at least one of these is not nice is bounded from above by
\begin{align*}
\sum_{i=0}^{k} (n/2^i)^{-1}
=
\frac{1}{n}
\sum_{i=0}^{k} 2^i
=\frac{2^{k}-1}{n}
<{2/\log n} = o(1),
\end{align*}
as required.
\end{proof}

The following lemma is a purely deterministic one, and is the geometric core of our argument.

\begin{lemma}
\label{lem:geometriccore}
Suppose that $(\log n) / 2 \leq t \leq n/2$ and that $R_t$ and $R_{2t}$ are nice.
Let $v$ be a vertex in the spanning component of $R_{2t}$. 
Then there exists a path of length $O(\log t)$ in $R_{2t}$ from $v$ to the spanning component of $R_{t}$.
%Then the diameter of the spanning component of $R_n$ is $O(\log^2 n)$.
\end{lemma}

%\am{Abbas: the proof below has been expanded a little  and some calculation steps are added}

\begin{proof}
Let $M \coloneqq \pi t r_t^2$, $W \coloneqq M r_t \log t$,
and consider the slabs of size $1\times W$ and $W \times 1$ defined in Definition~\ref{def:spanning}.
Since $R_t$ is nice, it has a spanning component, so each horizontal (vertical) slab
has a left-to-right (top-to-bottom) subsquare crossing.
Enumerate the slabs from left to right,
and from top to bottom.
Suppose $v$ lies in the $a$-th horizontal slab
and the $b$-th vertical slab.
%As in Definition \ref{def:spanning}, the small side of each slab is $ W=M r_t \log t$.

Let $k \coloneqq \left\lceil 2 + 3  \Gamma / (\pi t^2 r_t^4)\right\rceil$.
The left-to-right crossings of slabs
$a-k$ and $a+k$
and the top-to-bottom crossings of slabs
$b-k$ and $b+k$
constitute a cycle 
$u_1 u_2 \dots u_{\ell}$ in $R_t$, enclosing $v$,
such that for all $i=1,2,\dots,\ell$,
we have
\begin{enumerate}[(i)]
\item
$d_E(u_i,u_{i+1})\leq r_t/2$ (by the definition of a crossing)
and 
\item
$
(k-1) W\leq
d_E(v,u_i) \leq
2k W
$ .
\end{enumerate}
See Figure~\ref{fig:cycle}.
\begin{figure}
\begin{center}
\includegraphics[width=0.4\textwidth]{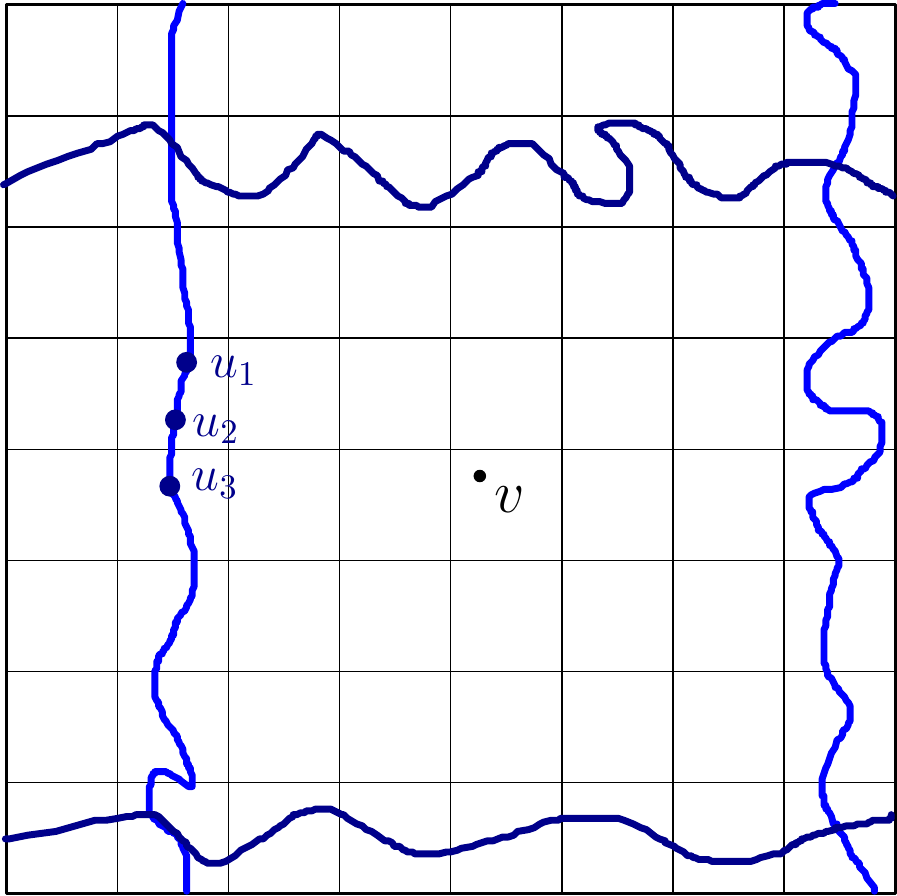}
\caption{Illustration for the proof of Lemma~\ref{lem:geometriccore}, with $k=3$: the crossings constitute a cycle $u_1u_2\dots$ enclosing the vertex $v$.}
\label{fig:cycle}
\end{center}
\end{figure}

We claim that there exists a path in $R_{2t}$ from $v$ to some $u_i$ in the cycle.
Let us show first that this claim completes the proof.
{Since $k> 2$, we have
\(k-1 > k/2 > 3 \Gamma / (2\pi t^2 r_t^4).\)
By definition of $W$ and $M$, we have $W = M r_t \log t = \pi t r_t^3  \log t$.
Also, $r_{2t}=r_{2t}/\sqrt{2}$ and $\log(2t) \leq 2 \log t$ for $t\geq2$.
Therefore,
$$\frac{\Gamma \log (2t)}{2 t r_{2t}}
\leq 
\frac{\sqrt 2 \Gamma \log t}{t r_{t}} <
\frac{3 \Gamma}{2\pi t^2 r_t^4} \times \pi t r_t^3  \log t
< (k-1)W \:.$$
Using property (ii) above, we find that 
$${\Gamma \log (2t)}/{(2 t r_{2t})} <(k-1)W \leq d_E(v,u_i)\:.$$
%\leq k(\frac{\sqrt{2}}{6})r_t M \log 2t\leq 
%(k-1) M r_t \log t  = 
%\leq d_E(v,u_i).$$}
%\[
%\Gamma =\left(\frac{A_2^2}{3\pi}\right) k =\left(\frac{ %t r_t^2M}{3}\right) k.
%\]
%\leq 2k M r_t \log t.
%\jj{I can get this to $\Gamma \log (2t) / (2 t r_{2t})\leq k(\frac{\sqrt{2}}{6})r_t M \log 2t$, probably this is smaller than $(k-1)Mr_t\log t$ but a few words of explanation would help..}
%Thus the conditions of Lemma \ref{lem:rgg} are satisfied for $R_{2t}$. 
Since $R_{2t}$ is nice, by property (ii) in Definition~\ref{def:nice} we have that
$$d_{R_{2t}}(v,u_i) \leq 
\frac{\eta d_E(v,u_i)}{r_{2t}} \leq
\frac{\eta 2 k W}{r_{2t}} =
O\left( \frac{ kW}{r_t} \right)
=
O\left( \frac{ k \pi t r_t^3  \log t}{r_t} \right)
=
O\left({ k t r_t^2  \log t} \right) = O(\log t)\:,$$
proving the lemma.
For the last equality we used the fact that $t r_t^2 = A_2/\pi = O(1)$ and also
$k \leq 3 + 3 \Gamma / \pi t^2 r_t^4
\leq 3 + 3 \pi \Gamma / A_2^2 = O(1)$.}

So it remains to show that,
there exists a path in $R_{2t}$ from $v$ to some $u_i$ in the cycle.
{Note that 
$$W = M r_t \log t = \pi t r_t^3  \log t = \pi t \log t \left( \frac{A_2}{t\pi}\right)^{3/2} = o(1)\:,$$
and since $k=O(1)$ as we saw above, we find that $kW = o(1)$.}
Since all vertices of the cycle lie within
distance
$2k W =o(1)$
of $v$,
and since $v$ lies in a spanning component of $R_{2t}$, there is a path $\zeta$ in $R_{2t}$ from $v$ to a vertex not enclosed by the cycle. 
Let $x_1$ be the last vertex in $\zeta$ enclosed by the cycle,
and let $x_2$ be the first vertex in $\zeta$ not enclosed by the cycle. 
If one of $x_1$ and $x_2$ is on the cycle,
then the claim is proved.
Otherwise, the line segment $x_1x_2$ crosses
some edge $u_i u_{i+1}$ of the cycle.
By the triangle inequality,
$$d_E(x_1,u_i)+d_E(x_2,u_{i+1})
\leq d_E(x_1,x_2)+d_E(u_i,u_{i+1})
\leq r_{2t} + r_t/2 \leq 2 r_{2t},$$
which implies that
at least one of $x_1 u_i$ 
and $x_2 u_{i+1}$ is an edge in $R_{2t}$.
Therefore, there exists a path in $R_{2t}$ from $v$ to at least one of $u_i$ and $u_{i+1}$, proving the claim.
\end{proof}

The following theorem about the size of the giant component  follows directly from Theorem~10.9 and Proposition 9.21 in \cite{penrose03}.

\begin{theorem}[\cite{penrose03}]
\label{thm:giant}
There exists a constant $a_{\textnormal{big}}$ so that \aas a random geometric graph $RGG(N,r)$  with $\pi N r  ^2  >a_{\textnormal{big}}$ has a connected component containing at least $0.99 N$ of its vertices.
\end{theorem}

Theorem \ref{thm:main_diam} now follows directly from the following theorem. 

\begin{theorem}
Let $G=G_n$ be a graph produced by the SPA model with parameters $A_1\in [0,1)$, $A_2>\max\{a_{\textnormal{big}},a_{\textnormal{large}}\},m=2,p=1$.
Then \aas $R_n$ has a giant component $C_n$ which contains at least $\sqrt{0.9} n $ vertices, and the diameter of $C_n$ is
$O(\log^2 n)$. 
Therefore, the effective diameter of $G_n$ is $O(\log^2 n)$.
\end{theorem}

\begin{proof}
Let $k$ be the smallest integer so that $n2^{-k}\leq \log n$.
On the one hand, by Lemma~\ref{lem:union},
a.a.s.\ we have that 
all random geometric graphs 
$R_n, R_{n/2},\dots, R_{n2^{-k}}$
are nice.
Let $C_n$ denote the spanning component of $R_n$.
%Then Lemma~\ref{lem:geometriccore} impliesthat the diameter of $C_n$ is $O(\log^2 n)$.
Let $y_0$ be a vertex in $C_n$. By repeated application of Lemma~\ref{lem:geometriccore}, \aas there exists a sequence of vertices $y_1,y_2,\dots ,y_k$, 
%where $k$ the smallest integer that $n2^{-k}\leq \log n$, 
with the following properties:
for each $i$, $0<i\leq k$, $y_i$ is in the spanning component of $R_{n2^{-i}}$, and there exists a path of length $O(\log n)$ from $y_i$ to $y_{i-1}$. 
Now $R_{n2^{-k}}$ has size at most $\log n$. Therefore, there exists a path of length $O(\log^2 n)$ between any two vertices in $C_n$.

On the other hand, by Theorem~\ref{thm:giant} and the fact that $R_n$ is distributed as $RGG(n,r_n)$ with $r_n=\sqrt{\frac{A_2}{\pi n}}$, we have that \aas $R_n$ contains a giant component which contains at least $0.99 n >\sqrt{0.9}n$ vertices. 
It is easy to see that this giant component must be the same as the spanning component $C_n$,
and this completes the proof.
\end{proof}

The methods used in this section do not suffice to show that the diameter of the giant component of $G_n$ is also logarithmic. In principle, it could be that there exist vertices in $G_n$ that are not contained in the giant component of $R_n$, but that are connected to this component by a long path that uses edges from inside the minor components of the graphs $R_n$, $R_{n/2}$, etc. 
Nevertheless, we believe that the SPA model graphs have logarithmic diameter inside their giant components a.a.s., and we leave this as an open problem.

\section{Lower bounds for rumour spreading}
\label{sec:gossip}
In this section we prove Theorem~\ref{thm:main_push}.
Recall that $m$ denotes the dimension and $n$ denotes the number of vertices.
We will first establish some structural properties of the graph generated by the SPA model,
and then use these to prove results about the rumour spreading protocols.
Let $c_m$ denote the volume of the $m$-dimensional ball of unit radius. 
{
The proof is based on a classification of edges according to their length, and vertices and edges according to their birth time.  

\begin{definition}\label{def:long}
 Let $\tau=\tau(n)=n^\beta$ and  $L=L(n)=n^{-\eta}$, where 
 $\eta, \beta \in (0,1)$ and
 %, and $y=y(n)$ and $\omega=\omega{n}$ be functions satisfying
\begin{equation}
\eta m <\beta (1-pA_1).\label{Lw}
%c_m L^m & > \tau^{pA_1 - 1} \log ^2 n,  \:, \mathrm{\ and}\label{Lw}\\
%\tau  & = n^{\beta} \:, \mathrm{\ for\ some\ } \beta \in (0,1).
%y & < n / \log n \:. \label{ynlogn}
\end{equation}
Say an edge is \emph{long} if the distance between its endpoints is larger than $L$, and is \emph{short} otherwise.
A vertex/edge is \emph{old} if it was born during one of the rounds $1,2,\dots,\tau$, and is \emph{new} otherwise. 
\end{definition}}

The following lemma establishes properties of old and new vertices and long and short edges. 

\begin{lemma}[Structural properties of the SPA model]
\label{spa_properties}
Let $G$ be a graph generated by the SPA model with $pA_1<1$.
Let $\varepsilon \in (0,1)$ be a constant independent of $n$. Let $\tau =n^{\beta}$ and $L=n^{-\eta}$ be as in Definition \ref{def:long}, and let $y=y(n)$ be a function satisfying
\begin{equation}
%c_m L^m & > w^{pA_1 - 1} \log ^2 n  \:, \label{Lw}\\
y  { =  n^{\gamma}  \:, \mathrm{\ for\ some\ }\gamma\in (\beta,1-\varepsilon/pA_1)}\label{yw}
%y \log n }&{\color{blue} =o(n^{1-\varepsilon/pA_1}) }\:. \label{ynlogn}
\end{equation}
A.a.s.\ we have the following properties.
\begin{enumerate}
\item[(a)]
All new edges are short.

\item[(b)]
%If there exists a constant $\theta\in(0,1)$ such that 
%\begin{equation}
%\tau \left( (\tau/y)^{\theta pA_1 / 2} + n^{-4/3}\log n\right) = o(1) \label{thetacond} \:,
%\end{equation}
If $\beta $ and $\gamma$ are such that
\begin{equation}
\beta< (\gamma-\beta)(pA_1)/2,\label{thetacond}
\end{equation} 
then all old vertices $v$ satisfy
\begin{equation}
\label{eq:degcond}
\frac{\deg(v,\tau)}{\deg(v,n)} <
 n^{\varepsilon}\left(\frac{y \log n}n\right)^{pA_1}\:.
\end{equation}
\end{enumerate}

\end{lemma}

{By part (a) of the lemma, all edges created after round $\tau$ are short, and thus the left hand side of~(\ref{eq:degcond}) gives an upper bound on the proportion of edges incident with vertex $v$ that are long. By (\ref{yw}), the right hand side of (\ref{eq:degcond}) is $o(1)$. }Therefore,  assertion~\eq{eq:degcond} quantifies the informal statement ``most edges incident to an old vertex are short.''
The proof of the above lemma is somewhat technical and 
can be found at the end of this section.

\begin{theorem}[Main Theorem for rumour spreading]
\label{main_rumourpp}
Let $\tau =n^\beta$ and $L=n^{-\eta}$ be as in Definition \ref{def:long}. Let $\varepsilon \in (0,1)$, and let $y=n^\gamma$ be a function satisfying \eq{yw} and~\eqref{thetacond}. Let $T=T(n)=n^\alpha$, where 
$\alpha$ is such that
\begin{equation}
%T \tau n^{\varepsilon - pA_1} (y \log n)^{pA_1}=o(1) 
{\alpha +\beta + (\varepsilon-pA_1)+\gamma pA_1 <0}.
\label{complicatedpp}\:
\end{equation}
Then, if the rumour starts from a uniformly random vertex,  a.a.s.\ after $T$ rounds of the push\&pull protocol, all informed vertices lie within distance $TL$ of the initial vertex.
\end{theorem}

\begin{proof}
Let $G$ be a graph generated by the SPA model,
and let $\zeta$ be a uniformly random vertex of it.
We may assume that $G$ satisfies the properties (a) and (b) given in Lemma~\ref{spa_properties}. { Note first that, if the rumour passes only through short edges, then in each time step the rumour can only spread to vertices that are within distance $L$ of any vertex with the rumour. Thus, in $T$ time steps the rumour can only reach vertices within distance $TL$ of the initial vertex.}

Let $B$ denote the bad event ``the rumour passes through a long edge during the first $T$ rounds.''
We need only show that a.a.s.\ $B$ does not happen. 
Note that new vertices are not incident to long edges by Lemma~\ref{spa_properties}(a).
Moreover, by Lemma~\ref{spa_properties}(b) every old vertex $v$ satisfies~\eq{eq:degcond}, which guarantees that most edges incident to $v$ are short.
Condition~\eq{eq:degcond} implies that the probability that an old vertex contacts a neighbour along some long edge in a given round is smaller than $n^{\varepsilon}\left(\frac{y \log n}n\right)^{pA_1}$.
There are exactly $\tau$ old vertices.
By the union bound over all old vertices and over the rounds $1$ to $T$, 
we find that 
%the probability that an old vertex contacts a neighbour along some long edge is bounded by
\begin{align*}
\p{B} \leq \tau T n^{\varepsilon}\left(\frac{y \log n}n\right)^{pA_1}
 & {=n^{\beta +\alpha+\varepsilon +(\gamma -1)pA_1}\log^{pA_1} n =o(1) }& \mathrm{by\ \eq{complicatedpp}. }\qquad\qquad \qedhere
\end{align*}
\end{proof}

%\begin{corollary}
%\label{cor2}
%Let $\delta>0$ be an arbitrarily small constant and assume that $m=2$ and $pA_1<1$.
%Define 
%$$\lambda := \frac{pA_1(1-pA_1)}{10+2pA_1}  \in (0,1) \:.$$
%If the rumour starts from a uniformly random vertex, then a.a.s.\ after $n^{\lambda-2\delta}$ rounds of the push\&pull protocol, the number of informed vertices is $O\left(n^{1-2\delta}\right)=o(n)$.
%\end{corollary}

%\am{Abbas: Definition of $\eps$ added in the following proof}

\begin{proof}[Proof of Theorem~\ref{thm:main_push}]
Set 
\begin{align*}
\delta& \coloneqq a(1-a)/K-\alpha > 0,\\
\tau & \coloneqq n^{m a / K},\\
y & \coloneqq n^{m(2+a)/K + \delta},\\
L & \coloneqq n^{-a(1-a)/K + \delta/2}, \textnormal{and }\\
T & \coloneqq n^{\alpha} = n^{a(1-a)/K-\delta},\\
\eps & \coloneqq \min \{am/K, \delta(1+a)\} / 2\:,
\end{align*}
and observe that \eq{Lw}, \eq{yw}, \eq{thetacond} and~\eq{complicatedpp} are satisfied, 
{ and that $TL=o(1)$}.\label{parameterpushpull} 
%\jj{I am having trouble verifying \eqref{yw}. What happened to $\varepsilon$? Should we define that too? I don't think the requirements hold for all $\varepsilon$.}
By Theorem~\ref{main_rumourpp},
a.a.s.\ after $T$ rounds of the push\&pull protocol,
all informed vertices lie in a ball of volume $O\left((TL)^m\right)=o(1)$.
By a standard Chernoff bound,
%~\eq{chernoff_upper_tail}, 
a.a.s.\ the number of vertices in any such ball is $O(n(TL)^m)=o(n)$.
\end{proof}

In the rest of this section we prove Lemma~\ref{spa_properties}.
We will use the following two concentration bounds, sometimes called multiplicative Chernoff bounds (see, e.g., \cite[Theorem 2.3(b,c)]{concentration}).
Let $X$ be a sum of independent indicator random variables and let $\delta\ge 0$. We have
\begin{equation}
\p{X \le (1-\delta) \e{X}} \le \exp(-\delta^2 \e{X}/2) 
\label{chernoff_lower_tail}
\end{equation}
and
\begin{equation}
\p{X \ge (1+\delta) \e{X}} \le
\exp\left(-\frac{\delta^2\e{X}}{2+2\delta/3}\right)
.
\label{chernoff_upper_tail}
\end{equation}

We will use the following theorem from~\cite{geoSPA}.

\begin{theorem}[Theorem~5.2 in~\cite{geoSPA}]
Let $f(n)$ be any function tending to infinity with $n$.
Let $v$ be a vertex with $\indeg(v,R)\geq f(n)\log n$.
Then, with probability at least $1-O(n^{-4/3})$, for all $r\in\{R,R+1,\dots,2R\}$ we have
$$\left| \indeg(v,r) - \indeg(v,R)(r/R)^{pA_1} \right| \leq \frac{2 r \sqrt{ \indeg(v,R) \log n}}{p A_1 R} \:.$$
\end{theorem}

In particular, setting $r=2R$, the above theorem implies that if $\indeg(v,R) \geq f(n) \log n$ for some $f(n)=\omega(1)$, then 
\begin{equation}
\label{doubling}
\p{\indeg(v,2R) \geq (2^{pA_1}-o(1))\indeg(v,R) %-\frac{4}{pA_1}\sqrt{\indeg(v,2R) \log n} 
} \geq 1 - O(n^{-4/3}) \:.
\end{equation}

\begin{lemma}
\label{lem:technical}
Let $\delta, \varepsilon\in (0,1)$ be arbitrary constants, and let $\tau=n^\beta$ and $y=n^\gamma$ be functions satisfying
\eq{yw}.
For any old vertex $v$ we have
$$\p{\frac{\indeg(v,\tau)}{\indeg(v,n)} \ge 
 n^{\varepsilon}\left(\frac{y \log n}n\right)^{pA_1}} = O\left((\tau/y)^{(1+o(1))\delta^2 
 pA_1 / 2} + \frac{\log n}{n^{4/3}}\right).$$
\end{lemma}

\begin{proof}
If $d := \indeg(v,\tau)=0$, the conclusion is obvious, so assume that $d\geq 1$.
 %Let $\mu := pA_1 >0$.
Define the following events:
\begin{align*}
E_1 & := \left\{\indeg(v,y) \ge (1-\delta) pA_1 d \log (y/\tau)\right\}, \\
E_2 & := \left\{\indeg(v,y\log n) \ge
(1-\delta)^2 (pA_1)^2 d \log (y/\tau) (\log \log n)\right\}, \mathrm{and}\\
E_3 & := \left\{\indeg(v,n) \ge 
(1-\delta)^2 (pA_1)^2 d \log (y/\tau) (\log \log n)
(n/(y \log n))^{pA_1} n^{-\varepsilon} \right\}.
\end{align*}

{ We remark that $\log (y/\tau) = (\gamma -\beta)\log n$.}
%, and thus $E_2$ is well defined by condition \eq{yw}.
Note that if $E_3$ happens then
$$\frac{d}{\indeg(v,n)} <
\left(\frac {y \log n}{n}\right)^{pA_1} n^{\varepsilon}\:.$$
{ Hence to prove the lemma we want to bound the probability of $E_3^c$.}
We will prove that
the probabilities $\p{E_1^c}$ and $\p{E_2^c|E_1}$
are at most $O((\tau/y)^{(1+o(1))\delta^2  pA_1 / 2})$,
and that $\p{E_3^c | E_1,E_2}$ is at most
$O(n^{-4/3}\log n)$.
This would prove the lemma, since
$$\p{E_3^c}
\le \p{E_1^c} + \p{E_3^c | E_1}
\le \p{E_1^c} + \p{E_2^c|E_1} +
\p{E_3^c | E_1,E_2} .$$

First, we bound $\p{E_1^c}$.
%Recall that $\mu = p A_1$.
Note that for each $i\in\{\tau+1,\dots,y\}$, the probability that $v_i$ creates an edge to $v$ is at least $pA_1 d /i$.
In fact, $\indeg(v,y)-\indeg(v,\tau)$ is stochastically larger than the sum of $y-\tau$ independent indicator variables $X_{\tau+1},\dots,X_y$
with $\e{X_i}= pA_1 d /i$,
as in this formula we have ignored the neighbours accumulated in rounds $\tau +1,\dots,y$.
Then 
$$\e{\sum X_i} =(1+o(1)) pA_1 d (\log(y/\tau)),$$
hence by the multiplicative Chernoff bound~\eq{chernoff_lower_tail},
$$\p{E_1^c}
\le \exp(-\delta^2 (1+o(1)) pA_1 d \log(y/\tau)/2) = (\tau/y)^{(1+o(1))pA_1 d \delta^2/2} \:.$$

Second, we bound $\p{E_2^c | E_1}$.
Conditional on $E_1$, 
by a similar argument, the difference in in-degrees
$\indeg(v,y \log n)-\indeg(v,y)$ is stochastically larger than the sum of $y \log n - y$ independent indicator variables $Y_{y+1},\dots,Y_{y\log n}$
with $\e{Y_i}= (1-\delta)  (pA_1)^2 d \log (y/\tau ) /i$.
Since
$$\e{\sum Y_i} = (1+o(1)) (1-\delta)  (pA_1)^2 d \log (y/\tau) (\log \log n) \:,$$
by the multiplicative Chernoff bound~\eq{chernoff_lower_tail},
\begin{align*}
\p{E_2^c | E_1}  \le \exp(- (1+o(1))\delta^2 (1-\delta)  (pA_1)^2 d \log (y/\tau) (\log \log n)/2)  = (\tau/y)^{-\Omega(d \log \log n)} \:.
\end{align*}

Finally, conditional on $E_1$ and $E_2$ we may use
\eqref{doubling} repeatedly for $R=y\log n, 2y\log n, \dots,$ all the way up to $R=n/2$ to obtain that with probability at least $1 - O((\log n ) n^{-4/3})$ we have
\begin{align}
\indeg(v,n) & \geq \left(2^{pA_1}-o(1)\right)^{\log_2(n/(y\log n))}\indeg(v,y\log n) \\& >
(n/(y\log n))^{pA_1}n^{-\varepsilon} \indeg(v,y\log n)\:,\label{question}
\end{align}
completing the proof.
%\jj{How do we get the step from the $o(1)$ term to $n^{-\varepsilon}$? This is the same $\varepsilon $ fixed in the theorem, right?}
Note that condition~\eq{yw} together with $E_2$ ensure that $\indeg(v,y\log n)=\omega(\log n)$ hence \eqref{doubling} can indeed by applied.
\end{proof}

%we have
%$\p{E_3^c | E_1,E_2} = O(n^{-4/3} \log n)$
%by using the following $O(\log n)$ times: \cite[Theorem~5.2]{geometricproperties}. 
%If a vertex $v$ has degree $d =\omega(\log n)$ at time $T$ then with  probability $1-O(n^{-4/3})$ 
%for all $T \le t \le 2T$ we have
%$$|\indeg(v,t) - d (t/T)^{pA_1}|
%\le \frac{2 t \sqrt{d\log n} }{pA_1 T}.
%$$

The following result follows  from the proof of  \cite[Theorem~1.5]{spa_def}.

\begin{theorem}
A.a.s.\ all vertices have outdegree $O(\log^2 n)$.
\end{theorem}

Since $\varepsilon$ can be chosen arbitrarily in Lemma~\ref{lem:technical}, and all outdegrees are polylogarithmic a.a.s., we may replace indegrees with total degrees in Lemma~\ref{lem:technical} and conclude the following.

\begin{corollary}
\label{cor:technical}
Let $\tau=\tau(n)=n^\beta$ and $y=y(n)=n^\gamma$ be functions satisfying
\eq{yw}, and let $\delta,\varepsilon\in(0,1)$ be arbitrary constants.
For any old vertex $v$  we have
$$\p{\frac{\deg(v,\tau )}{\deg(v,n)} \ge 
 n^{\varepsilon}\left(\frac{y \log n}n\right)^{pA_1}} = O\left((\tau /y)^{(1+o(1))\delta^2 pA_1 / 2} + \frac{\log n}{n^{4/3}}\right).$$
\end{corollary}

We will use the following theorem from~\cite{spa_typical}.

\begin{theorem}[Theorem~2.3 in \cite{spa_typical}]
\label{thm:degs}
Let $f(n)$ be any function that goes to infinity with $n$.
A.a.s.\ for all $i\in\{1,2,\dots,n\}$
and all $t\in\{i,\dots,n\}$ we have
$$\indeg(v_i,t) = O \left( 
f(n) (\log n)  (t/i)^{pA_1}
\right) \:.$$
\end{theorem}

We now have all the ingredients to prove
Lemma~\ref{spa_properties}.

%jj: this still has to be checked. 		
\begin{proof}[Proof of Lemma~\ref{spa_properties}]
(a)
By Theorem~\ref{thm:degs}, a.a.s.\  
for all $i\in[n]$
and all $t\in\{i,\dots,n\}$ we have
$$\indeg(v_i,t) \le (\log n)^{3/2} (t/i)^{pA_1}.$$
Suppose this is the case and let $i\in[n]$.
Then at any time $t \in \{\max\{i,\tau +1\},\dots,n\}$, the sphere of influence of $v_i$ at time $t$ has volume at most
\begin{align*}
\frac{A_1 (\log n)^{3/2} (t/i)^{pA_1}+A_2}{t}
& \le
A_1 (\log n)^{3/2}\: t^{pA_1-1}+\frac{A_2}{t} \\
& < \tau^{pA_1 - 1} \log ^2 n & \mathrm{since\ }pA_1<1\\
& < c_m L^m\:. & \mathrm{by\ \eq{Lw}}
\end{align*}
Thus any incoming edge to $v_i$ that is created after round $w$ is short.

(b)
The number of old vertices is $\tau=n^{\beta}$, and the probability that an old vertex
fails to satisfy~\eq{eq:degcond} is 
$O((\tau/y)^{(1+o(1))\delta^2 pA_1 / 2}) + O(n^{-4/3}\log n)$
for any constant $\delta\in(0,1)$
by Corollary~\ref{cor:technical}.
Recall that $(\tau/y)=n^{-(\gamma-\beta)} $, where $\beta <\gamma$. By \eq{thetacond} 
we can choose $\delta$ close enough to 1 so that  $(\gamma-\beta)(1+o(1))\delta^2 pA_1 / 2>\beta$. 
The union bound completes the proof.
\end{proof}

\bibliographystyle{plain}
\bibliography{RumourSPA}

\end{document}